\documentclass[12pt,transaction, draftclsnofoot,onecolumn]{IEEEtran}
\usepackage{amsfonts}
\usepackage{amssymb}
\usepackage[bookmarks=false, draft]{hyperref}
\usepackage{graphicx}
\usepackage{subfigure}
\usepackage{enumerate}
\usepackage{amsmath}
\usepackage{color}
\usepackage{amsthm}
\usepackage{amsmath}
\usepackage{algorithm}
\usepackage{algpseudocode}
\usepackage{breakurl}
\hypersetup{hidelinks}

\newcommand{\bp}{\begin{proof} \small }
\newcommand{\ep}{\end{proof} \normalsize}
\newcommand{\epx}{\end{proof} \small}
\newcommand{\bpa}{\begin{proofappx} \footnotesize }
\newcommand{\epa}{\end{proofappx} \small }
\newtheorem{theorem}{Theorem}
\newtheorem{proposition}{Proposition}

\newtheorem*{theorem*}{Theorem}
\newtheorem*{proposition*}{Proposition}
\newtheorem*{corollary*}{Corollary}
\newtheorem*{lemma*}{Lemma}
\newtheorem*{assumption*}{Assumption}
\newtheorem*{definition*}{Definition}
\newtheorem*{claim*}{Claim}

\newcommand{\be}{\begin{equation}}
\newcommand{\ee}{\end{equation}}
\newcommand{\bs}{\begin{subequations}}
\newcommand{\es}{\end{subequations}}
\newcommand{\bq}{\begin{eqnarray}}
\newcommand{\eq}{\end{eqnarray}}
\newcommand{\bqn}{\begin{eqnarray*}}
\newcommand{\eqn}{\end{eqnarray*}}

\newcommand{\ba}{\left[ \begin{array}}
\newcommand{\ea}{\\ \end{array} \right]}
\newcommand{\ben}{\begin{enumerate}}
\newcommand{\een}{\end{enumerate}}

\def\real{{\mathchoice%
{\hbox{\rm\setbox1=\hbox{I}\copy1\kern-.45\wd1 R}}
{\hbox{\rm\setbox1=\hbox{I}\copy1\kern-.45\wd1 R}}
{\hbox{\scriptsize\rm\setbox1=\hbox{I}\copy1\kern-.45\wd1 R}}
{\hbox{\scriptsize\rm\setbox1=\hbox{I}\copy1\kern-.45\wd1 R}}}}

\def\Zint{{\mathchoice{\setbox1=\hbox{\sf Z}\copy1\kern-.75\wd1\box1}
{\setbox1=\hbox{\sf Z}\copy1\kern-.75\wd1\box1}
{\setbox1=\hbox{\scriptsize\sf Z}\copy1\kern-.75\wd1\box1}
{\setbox1=\hbox{\scriptsize\sf Z}\copy1\kern-.75\wd1\box1}}}
\newcommand{\complex}{ \hbox{\rm C\kern-0.45em\rule[.07em]{.02em}{.58em}%
\kern 0.43em}}

\begin{document}
%
\title{E\textsuperscript{2}M\textsuperscript{2}: Energy Efficient Mobility Management in Dense Small Cells with Mobile Edge Computing}

\author{\IEEEauthorblockN{Jie Xu$^*$, Yuxuan Sun$^\dagger$, Lixing Chen$^*$, Sheng Zhou$^\dagger$}
\IEEEauthorblockA{$^*$Department of Electrical and Computer Engineering, University of Miami, USA\\
$^\dagger$Department of Electronic Engineering, Tsinghua University, China}}

\maketitle

\begin{abstract}
Merging mobile edge computing with the dense deployment of small cell base stations promises enormous benefits such as a real proximity, ultra-low latency access to cloud functionalities. However, the envisioned integration creates many new challenges and one of the most significant is mobility management, which is becoming a key bottleneck to the overall system performance. Simply applying existing solutions leads to poor performance due to the highly overlapped coverage areas of multiple base stations in the proximity of the user and the co-provisioning of radio access and computing services. In this paper, we develop a novel user-centric mobility management scheme, leveraging Lyapunov optimization and multi-armed bandits theories, in order to maximize the edge computation performance for the user while keeping the user's communication energy consumption below a constraint. The proposed scheme effectively handles the uncertainties present at multiple levels in the system and provides both short-term and long-term performance guarantee. Simulation results show that our proposed scheme can significantly improve the computation performance (compared to state of the art) while satisfying the communication energy constraint.
\end{abstract}


%
\IEEEpeerreviewmaketitle
\section{Introduction}
Mobile networks are experiencing a paradigm shift. A multitude of new technologies are being developed to accommodate the surging data demand from mobile users. Among these technologies, Ultra Dense Networking (UDN) \cite{quek2013small}\cite{chih2014toward} and Mobile Edge Computing (MEC) (a.k.a. Fog Computing) \cite{shi2016edge}\cite{roman2016mobile} are considered as key building blocks for the next generation mobile network (5G). UDN increases the network capacity through the ultra-dense deployment of small cell Base Stations (BSs), which is viewed as the key technology to realize the so-called 1000x capacity challenge \cite{chen2014requirements}. MEC provides access to cloud-like computing and storage resources at the edge of the mobile network, creating significant benefits such as ultra-low latency and precise location awareness, which are necessary for emerging applications such as Tactile Internet, Augmented Reality, Connected Cars and the Internet of Things. In many ways, UDN may provide the strongest case for MEC since small cell BSs are inherently edge-oriented. It is envisioned that small cell BSs endowed with cloud functionalities will be a major form of MEC service provision nodes \cite{scf2015}.

Although there are increasingly many works studying UDN and MEC, they are mostly separate efforts. However, significant new challenges are created by the envisioned integration of UDN and MEC, despite the enormous potential benefits. A key bottleneck to the overall system performance in this new paradigm is mobility management, which is the fundamental function of tracking mobile User Equipments (UEs) and associating them with appropriate BSs, thereby enabling mobile service to be delivered. Traditionally, mobility management was designed for macro cellular networks with infrequent handovers between cells, and providing radio access service only. Merging UDN and MEC drastically complicates the problem. Simply applying existing solutions leads to poor mobility management due to the highly overlapped coverage areas of multiple BSs in the proximity and the co-provisioning of radio access and computing services. In particular, what impedes efficient mobility management is the tremendous uncertainty at multiple levels:
\begin{itemize}
\item \textbf{Present uncertainties}. A key challenge is the unavailability of accurate information of candidate BSs in the proximity. Had the system know a priori which BSs offers the best performance (either throughput, energy consumption or computation latency), the mobility management decision making can be much easier, thus avoiding frequent handovers which lead to energy inefficiency.
\item \textbf{Future uncertainties}. A perhaps even bigger challenge is that the long-term energy budget couples the mobility management decisions across time, and yet the decisions have to be made without foreseeing the future (i.e. the future locations, channel conditions, available edge computing resources etc.)
\end{itemize}

\begin{figure}
  \centering
  \includegraphics[width=0.48\textwidth]{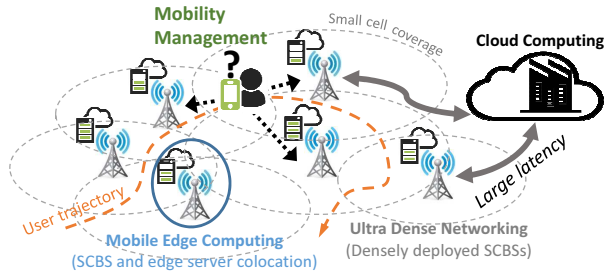}\\
  \caption{Mobility management in small cell-based MEC.}\label{system}
\end{figure}

Figure \ref{system} illustrates the mobility management problem in small cell-based MEC. In this paper, we develop a novel user-centric mobility management scheme, called E\textsuperscript{2}M\textsuperscript{2} (Energy-Efficient Mobility Management), that (i) smartly learns the best BSs to connect and makes proper handover decisions in an online fashion, and (ii) satisfies both short-term and long-term performance requirement. This is mainly achieved by leveraging and extending Lyapunov Optimization \cite{neely2010stochastic} and Multi-armed Bandits \cite{auer2002finite}. In particular, we provide strong performance guarantee by employing our scheme: (1) the computing performance is within a bounded deviation of the optimal performance that can be achieved by an imaginary oracle that knows all information in advance; (2) the long-term energy consumption constraint is approximately satisfied within a bounded deviation.

\section{Related Work}
Mobile edge computing has received an increasing amount of attentions in recent years. A central theme of many prior studies is offloading policies, i.e. what/when/how to offload a user's workload from its device to the edge system or cloud (see \cite{huang2012dynamic}\cite{satyanarayanan2009case} and references therein). Our paper does not study offloading policies and can work in conjunction with any offloading policy. There is also extensive work on dense small cell networks, mainly focusing on radio resource management and interference mitigation (see \cite{shen2015silence}\cite{torregoza2010joint} and references therein). Our paper is also orthogonal to this literature and does not impose any assumption on how the network is planned or how the radio resource is allocated among the cells (hence we allow inter-cell interference in our model).

Much of the mobility management research has been based on optimization theory \cite{ye2013user}\cite{mesodiakaki2014energy}\cite{bottai2014energy}\cite{althunibat2014handover}. For instance, a utility maximization problem is formulated for the optimal user association \cite{ye2013user}. A handover policy is proposed in \cite{althunibat2014handover} in which low energy efficiency from the serving BS triggers a handover. These solutions have been proved effective for less-densified heterogeneous networks, but they may perform poorly when the network density becomes high as frequent handovers may occur. To address this challenge, a learning-based mobility management scheme is proposed in \cite{shen2016non} based on the multi-armed bandits technique \cite{auer2002finite}. This scheme is user-centric, which has been an emerging trend of designing mobility management, particularly for the future 5G standard \cite{boccardi2014five}. However, all these works merely consider the radio access aspect of BSs. Endowing BSs with MEC capabilities requires new mobility management solutions.

To deal with future uncertainties, we leverage Lyapunov optimization \cite{neely2010stochastic},which has wide applications in communication networks and queueing systems. Our scheme builds upon yet extends Lyapunov optimization by integrating Multi-armed Bandits \cite{auer2002finite} to solve the optimization task in each time period due to the lack of accurate system information.

\section{System Model}
\subsection{Network Model}
We consider an ultra-dense cellular network with $N$ Base Stations (BS) (which can be macro BSs or small cell BSs), indexed by $\mathcal{N}=\{1,2,...,N\}$.  The network is divided into $M$ disjoint regions, indexed by $\mathcal{M}=\{1,2,...,M\}$.  Each region $m \in \mathcal{M}$ is covered by a subset of BSs, denoted by $\mathcal{A}(m) \subseteq \mathcal{N}$. We consider a representative mobile user that is moving in the network for $T$ time periods. In period $t$, the location of the user is denoted by $L^t \in \mathcal{M}$. Therefore, $\mathcal{L} = L^1\to L^2 \to ...\to L^T$ is the trajectory of the user.  This trajectory can be generated by any mobility model, such as the random waypoint model, and we do not require any prior knowledge about the trajectory.

At each location $L^t$ in time period $t$, the mobility management system makes decisions on which BS the user is served by, among the discoverable BSs $\mathcal{A}(L^t)$ in the vicinity.  The mobility decision is traditionally made at the BS (e.g. X2 handover in LTE) or Evolved Packet Core (e.g. S1 handover at the Mobility Management Entity in LTE). Recently, there has been an emerging trend of designing user-centric mobility management, particularly for the future 5G standard \cite{boccardi2014five}. In this work, we consider user-centric mobility management and let the user make mobility decisions.

\subsection{Workload and Service Model}
In each period $t$, the user has computation workloads with arrival rate $\lambda^t \in [0, \lambda_\text{max}]$ that needs to be offloaded to the mobile edge server for processing, where $\lambda_\text{max}$ is the maximum possible arrival rate. We assume that each workload has an equal size $\gamma$. As assumed in prior work, $\lambda^t$ is available at the beginning of each time period $t$. We focus on delay-sensitive interactive workloads, which are the main application scenario of MEC.

Each BS is co-located with an edge server. The maximum service rate of BS $n$ in time period $t$ is denoted by $s^t(n)$. The total workload offloaded to BS $n$ by other users is denoted by $\mu^t(n)$, which we refer to as background workload.  By offloading $\lambda^t$ to a BS $n \in \mathcal{A}(L^t)$,  the delay cost incurred to the representative user can be captured by a general function $d(\lambda^t, \mu^t(n), s^t(n))$. As a concrete example, we model the service process at each edge server as a M/G/1/PS (Memoryless/General/1/Processor-Sharing) queue and the average response time (multiplied by the arrival rate) to represent the delay cost. Specifically, the delay cost in period $t$ can be written as \cite{ren2013coca}
\begin{align}
d(\lambda^t, \mu^t(n), s^t(n)) = \frac{\lambda^t}{s(n) - (\mu^t(n) + \lambda^t)}
\end{align}
in which we ignore the transmission delay which is typically small since the BS is in close proximity.  While the M/G/1/PS queuing model may not capture the exact response time in practice, it has been widely used as an analytic vehicle to provide a reasonable approximation for the actual service process. It is also possible that the workload is further offloaded a remote cloud or distributed among a set of edge servers. Our framework is able to capture this scenario by using a proper delay cost function $d(\cdot)$ that absorbs this additional delay. For simplicity, this paper focuses on processing of computation workload on local edge servers.

\subsection{Channel and Power Consumption Model}
Since the considered period is relatively long, we consider only slow fading for the wireless channel for uploading the workload to the edge server. Let $H^t(n, m)$ denote the wireless channel state between region $m$ and BS $n \in \mathcal{A}_m$ in period $t$. Given the transmission power $P_{tx}$ of the user, the maximum achievable uplink transmission rate is given by the Shannon's theorem,
\begin{align}
r(H^t(n, m)) = W\log_2\left(1 + \frac{P_{tx}H^t(n, m)}{\sigma^2 + I^t(n)}\right)
\end{align}
where $W$ is the channel bandwidth and $\sigma^2$ is the noise power and $I^t(n)$ is the possible inter-cell interference in time period $t$. The required transmission energy for uploading workload with arrival rate $\lambda^t$ to BS $n$ in a period is therefore
\begin{align}
e(\lambda^t, H^t(n, L^t)) = \frac{P_{tx} \lambda^t  \gamma}{W\log_2\left(1 + \frac{P_{tx}H^t(n, L^t)}{\sigma^2  + I^t(n)}\right)}
\end{align}


\section{Problem Formulation}
We first formulate the problem by assuming that the user has complete information of the system for all time periods $t = 0, 1, ...,T-1$ in advance at time $t = 0$. Therefore, the user essentially is solving an offline optimization problem. However, in the real system, information in period $t$ can only be known in period $t$. Therefore, we will subsequently develop online algorithms and show that they are efficient compared to the optimal offline algorithm.  For the online algorithms, we will consider two practical scenarios depending on the availability of information. The information in each period $t$ can be classified into two categories:
\begin{itemize}
\item \textbf{User-side State Information}: The user's location $L^t$, the available candidate BSs $\mathcal{A}^t$ and the workload rate $\lambda^t$.
\item \textbf{BS-side State Information}: For each $n \in \mathcal{A}^t$, its background workload $\mu^t(n)$, its maximum service rate $s^t(n)$ and the uplink channel condition $H^t(n, L^t)$.
\end{itemize}
The first deployment scenario assumes that the user has both the user-side information and the BS-side information.  The second deployment scenario assumes that the user has only the user-side information.

\subsection{Offline Problem Formulation}
The user performs mobility management to maximize its computing performance (i.e. to minimize the delay cost) subject to an energy consumption budget (e.g. since battery is limited). Formally, the problem is as follows
\begin{align}
\textbf{P1:}&~~\min_{a^0,...,a^{T-1}} ~\frac{1}{T}\sum_{t=0}^{T-1} d(\lambda^t, \mu^t(a^t), s^t(a^t))\\
\text{s.t.} &~~\sum_{t=0}^{T-1} e(\lambda^t, H^t(a^t, L^t)) \leq \alpha B \label{budget}\\
&~~ r(H^t(a^t, L^t)) \geq r_\text{min}, \forall t \label{minrate}\\
&~~d(\lambda^t, \mu^t(a^t), s^t(a^t)) \leq d_\text{max}, \forall t \label{maxdelay}\\
&~~a^t \in \mathcal{A}(L^t) \label{coverage}
\end{align}
The objective is to minimize the average delay cost, which is the key performance metric for MEC applications which require real-time interaction, by deciding which BS the user should be served by in each time period.  The first constraint \eqref{budget} states that the total energy consumption is limited by the available battery during the current trip of the user, where $\alpha \in (0, 1]$ indicates the desired capping of energy consumption relative to the total battery capacity $B$. The second constraint \eqref{minrate} requires that the uploading transmission rate meets a minimum rate $r_\text{min}$ so that the transmission delay is negligible. The third constraint \eqref{maxdelay} requires that the per-slot delay does not exceed an upper limit $d_\text{max}$ so the real-time performance is guaranteed in the worst case. The last constraint \eqref{coverage} states that the candidate BSs are those that cover  $L^t$.

The first major practical challenge that impedes derivation of the optimal solution to \textbf{P1} is the lack of future information: optimally solving \textbf{P1} requires complete offline information (including the user-side and the BS-side information) over the entire trip periods that is extremely difficult, if not impossible, to accurately predict in advance. Furthermore, \textbf{P1} belongs to integer nonlinear programming and is difficult to solve, even if the long-term future information is accurately known a priori. Thus, these challenges demand an online approach that can efficiently make mobility management decisions without foreseeing the far future.

\subsection{$J$-Step Lookahead Algorithm}
Because the system state information (i.e. the location, the workload, the channel state etc.) can follow an arbitrary sample path, we use the $J$-step lookahead algorithm as an \textbf{benchmark} to evaluate our online algorithms. Specifically, we divide the entire trip duration into $R \geq 1$ frames, each having $J \geq 1$ time periods such that $T = RJ$. For the $r$-th frame, we define $D^*_r$ as the optimal delay cost associated with the following static optimization problem.
\begin{align}
\textbf{P2:}&~~\min_{a^{rJ},...,a^{(r+1)J-1}}~ \frac{1}{J}\sum_{t=rJ}^{(r+1)J -1} d(\lambda^t, \mu^t(a^t), s^t(a^t))\\
\text{s.t.} &~~\sum_{t=rJ}^{(r+1)J -1} e(\lambda^t, H^t(a^t, L^t)) \leq \frac{\alpha B}{R} \label{budgetJ}\\
&~~\text{constraints \eqref{minrate}, \eqref{maxdelay}, \eqref{coverage}}
\end{align}
The value $D^*_r$ thus represents the optimal empirical average delay cost for frame $r$ over all policies that have full knowledge of the future system state information over the frames and that satisfy the constraints. We assume that for every frame, there exists at least one sequence of (possibly randomized) mobility management decisions that satisfy the constraints of \textbf{P2}. Let $a^{rJ,*}, ...a^{(r+1)J-1,*}$ denote the sequence of the optimal mobility management decisions that achieve $D^*_r$. The minimum long-term average delay cost achieved by the oracle's optimal $J$-step lookahead algorithm is thus given by $D^*=\frac{1}{R}\sum_{r=0}^{R-1}D^*_r$.

\section{Online Mobility Management}
In this section, we develop online mobility management algorithms that do not require future system state information. We consider two deployment scenarios depending on whether the user has the exact BS-side state information.

\subsection{E\textsuperscript{2}M\textsuperscript{2} Algorithm with Full State Information}
We first consider the case in which the user has full state information in each period. A significant challenge of directly solving \textbf{P1} is that the long-term energy budget couples the mobility management decisions across different time periods: using more energy at the current time will potentially reduce the energy budget available for future uses, and yet the decisions have to be made without foreseeing the future. To address this challenge, we leverage Lyapunov optimization and construct a virtual energy deficit queue to guide the mobility management decisions. Specifically, let $q(0) = 0$, the dynamics of the energy deficit queue evolves as
\begin{align}\label{queue}
q(t+1) = \{q(t) + e(\lambda^t, H^t(a^t, L^t)) - \alpha B/T\}^+
\end{align}
The virtual queue length $q(t)$ indicates how far the current energy usage deviates from the battery energy budget.  The algorithm is shown in Algorithm 1.

E\textsuperscript{2}M\textsuperscript{2} is an online algorithm because it requires only the currently available information as the inputs.  We use $V_0, V_1, ..., V_{R-1}$ to denote a sequence of positive control parameters to dynamically adjust the tradeoff between delay cost minimization and energy consumption over the $R$ frames, each having $J$ periods. Lines 2 - 4 reset the energy deficit virtual queue at the beginning of each frame. Line 5 defines an online optimization problem \textbf{P3} to decide the mobility management decisions.  The optimization problem aims to minimize a weighted sum of the delay cost and energy consumption where the weight depends on the current energy deficit queue length. A large weight will be placed on the energy consumption if the current energy deficit is large. The energy deficit queue maintained without foreseeing the future guides the mobility management decisions towards meeting the battery energy constraint, thereby enabling online decisions.

\begin{algorithm}
\caption{E\textsuperscript{2}M\textsuperscript{2} with Full State Information (FSI)}
\begin{algorithmic}[1]
\State \textbf{Input}: $L^t$, $\mathcal{A}^t$, $\lambda^t$, $h^t$, and $\forall n \in \mathcal{A}^t$, $\mu^t(n)$, $s(n)$, $H^t(n, L^t)$ at the beginning of each $t$.
\If{$t = rJ, \forall r = 0,1,...,R-1$}
\State $q(t) \leftarrow 0$ and $V \leftarrow V_r$
\EndIf
\State Choose $a^t$ subject to \eqref{minrate}, \eqref{maxdelay}, \eqref{coverage} to minimize
\begin{align}
(\textbf{P3})~~V\cdot d(\lambda^t, \mu^t(a^t), s(a^t)) + q(t)\cdot e(\lambda^t, H^t(a^t, L^t)) \nonumber
\end{align}
\State Update $q(t)$ according to \eqref{queue}.
\end{algorithmic}
\end{algorithm}

\subsection{E\textsuperscript{2}M\textsuperscript{2} Algorithm with Partial State Information}
In many deployment scenarios, the user knows only the user-side information but not the BS-side information.  This creates a major difficulty in solving \textbf{P3} exactly in each period $t$.  For notational convenience, we define $d(n) = d(\lambda^t, \mu^t(n), s(n))$ and $e(n) = e(\lambda^t, H^t(n, L^t))$ for all $n \in \mathcal{A}(L^t)$, and $Z(n) = V\cdot d(n) + q(t)\cdot e(n)$. In each period $t$, the optimal BS for serving the user is thus
\begin{align}
a^{t,*} = \arg\min_n \{Z(n)\}
\end{align}
Since $\mu^t(n)$, $s^t(n)$ and $H^t(n, L^t)$ for any $n \in \mathcal{A}(L^t)$ are unknown by the user, the values of $d(n)$ and $e(n)$ and hence $Z(n)$ cannot be evaluated exactly. In this section, we augment our E\textsuperscript{2}M\textsuperscript{2} algorithm with an online learning algorithm that learns $a^{t,*}$ without requiring the exact BS-side information.

In order to learn the optimal mobility management decision, we divide each period $t$ into $K$ slots.  In each time slot within a time period, the learning algorithm connects the user with one BS $a^t_k \in \mathcal{A}(L^t)$. At the end of the time slot $k$, the user observes the delay $\tilde{d}^t_k$ and energy consumption $\tilde{e}^t_k$, which serve as feedback information to guide the learning of the optimal BS. Let $z^t_k = V\cdot \tilde{d}^t_k + q(t)\cdot \tilde{e}^t_k$. If $z^t_k$ were exactly $\frac{1}{K} Z(a^t)$, then learning the optimal $a^{t,*}$ would be simple:  use each $n \in \mathcal{A}(L^t)$ for one slot in a round-robin fashion and keep connecting to the one that minimizes $z^t_k$ for the remaining $K - |\mathcal{A}(L^t)|$ slots. However, due to the variance in workload arrivals and channel states, $z^t_k$ is only a noisy version of $\frac{1}{K} Z(a^t)$. As a result, the above simple learning algorithm can be very suboptimal since we may get trapped in a BS that results in a large $Z(n)$. In fact, this problem is a classic sequential decision making problem that involves a tradeoff between exploration and exploitation -- the learning algorithm needs to explore the different BSs to learn good estimates of $Z(n), \forall n$ while at the same time trying to connect to the optimal BS as long as possible. This problem is extensively studied under the multi-armed bandits framework and many learning algorithms have been developed with performance guarantee. In this paper, we augment our  E\textsuperscript{2}M\textsuperscript{2} algorithm with the widely adopted UCB1 algorithm \cite{auer2002finite} to learn the optimal BS under uncertainty. UCB1 is an index-based algorithm which, in each slot $k$, selects the BS with the largest index. The index for BS $n$ is an upper confidence bound on the empirical estimate of $Z(n)$. Nevertheless, other learning algorithms can also be incorporated in our framework.

The algorithm is shown in Algorithm 2. The major difference from Algorithm 1 is that instead of solving \textbf{P3} exactly, we use the UCB1 algorithm to learn the optimal BS to minimize the objective in \textbf{P3}, which is reflected in Lines 5 through 13.

\begin{algorithm}
\caption{E\textsuperscript{2}M\textsuperscript{2} with Partial State Information (PSI)}
\begin{algorithmic}[1]
\State \textbf{Input}: $L^t$, $\mathcal{A}^t$, $\lambda^t$, $h^t$ at the beginning of each $t$.
\If{$t = rJ, \forall r = 0,1,...,R-1$}
\State $q(t) \leftarrow 0$ and $V \leftarrow V_r$
\EndIf
\For{Each $n\in\mathcal{A}(L^t)$}
\State $\bar{z}^t(n) \leftarrow 0, \theta^t(n) \leftarrow 0$
\EndFor
\For{$k = 1,...,K$}    \Comment{\textit{UCB1 Learning}}
\State Connect to $n^{*}_k = \arg\min_{n} \bar{z}^t(n) - \sqrt{\frac{\alpha\ln k}{\theta^t(n)}}$
\State Observe $\tilde{d}^t_k$ and $\tilde{e}^t_k$
\State Update $\bar{z}^t(n^*_k) \leftarrow \frac{\theta(n^*_k)\bar{z}^t(n^*_k) + V\tilde{d}^t_k + q(t)\tilde{e}^t_k}{\theta^t(n^*_k)+1}$
\State Update $\theta^t(n^*_k) \leftarrow \theta^t(n^*_k) + 1$
\EndFor
\State Update $q(t)$ according to \eqref{queue}.
\end{algorithmic}
\end{algorithm}

\section{Performance Analysis}
We only characterize the performance of the E\textsuperscript{2}M\textsuperscript{2} with PSI because E\textsuperscript{2}M\textsuperscript{2} with FSI can be considered as a special case.

First, we characterize the performance of the UCB1 algorithm, which is used to approximately solve \textbf{P3}. Let $\tilde{Z}^t = \sum_{k=1}^K \tilde{z}^t_k$ be the empirical total weighted sum of delay cost and energy consumption obtained by running the UCB1 algorithm. Let $Z^{t,*} = Z(a^{t,*})$ be the optimal value with full information.

\begin{proposition}
If each period is divided into $K$ slots and let $\alpha = 2 (z^t_{max})^2$, then
\begin{align}
&\mathbb{E}\tilde{Z}^t - Z^{t,*}  \nonumber\\
\leq &z^t_{max}\left[8\sum_{n\neq a^{t,*}}\left(\frac{\ln K}{\delta^t(n)}\right) + \left(1+ \frac{\pi^2}{3}\right)\sum_{n\neq a^{t,*}}\delta^t(n) \right]
\end{align}
\end{proposition}
where $\delta^t(n) = (Z^t(n) - Z^{t,*})/K$ and $z^t_{max} = \sup z^t_k$.
\begin{proof}
This is based on the regret analysis of UCB1 \cite{auer2002finite}.
\end{proof}
Proposition 1 states that our learning algorithm approximately solves \textbf{P3} with a bounded deviation. Therefore, there exists a constant $C$  such that $\tilde{Z}^t - Z^{t,*} \leq C, \forall t$.

Next, we characterize the performance of E\textsuperscript{2}M\textsuperscript{2}  with PSI.

\begin{theorem}
For any $J \in \mathbb{Z}_+$ and $R \in \mathbb{Z}_+$ such that $T = RJ$, the following statements hold.

(1) The average delay cost achieved by E\textsuperscript{2}M\textsuperscript{2} with PSI satisfies
\begin{align}
d^* \leq \frac{1}{R}\sum_{r=0}^{R-1}D^*_r + \frac{U(J+1) + C}{R}\sum_{r=0}^{R-1}\frac{1}{V_r}
\end{align}
where $U$ is a finite constant.

(2) The energy consumption constraint is approximately satisfied with a bounded deviation:
\begin{align}
&\sum_{t=0}^{T-1}e(\lambda^t, H^t(a^t, L^t)) \nonumber\\
\leq & \alpha B + \sum_{r=0}^{R-1} \sqrt{2UJ(J+1) + CJ + V_rJD^*_r}
\end{align}
\end{theorem}
\begin{proof}
See Appendix.
\end{proof}

Theorem 1 proves a strong performance guarantee of our mobility management scheme: the computing performance in terms of delay cost is within a bounded deviation of the optimal performance that can be achieved by an imaginary oracle that knows future information and has the computational power the solve the offline optimization problem, while the long-term energy consumption constraint is approximately satisfied within a bounded deviation.

\section{Simulations}
In this section, we carry out simulation studies to validate our analysis and evaluate the performance of the proposed mobility management scheme.

\subsection{Simulation Setup}
We simulate a 1000m$\times$1000m square area in which 25 BSs are deployed on a regular grid network. The distance between two adjacent BSs is 160m. We used an adjusted random walk model to generate user trajectories. The mobility model is adjusted to avoid frequent moving back and forth in order to better capture the real world scenario. Since BSs are densely deployed, the user is able to observe multiple candidate BSs and associate to any one of them. The association radius is set to be 250m. Each time period is 5 minutes. The arrival rate of the user computational workload is set to be $\lambda^t \in [0, 12]$ per period. The edge server service rate is $s(n) = 50$ per period for all $n$ and the background workload is $\mu(n) \in [0, 40]$ per period. The wireless channel state $H^t$ is modeled using a pathloss model $P = 25.3 + 37.6\times\log d$, which is suggested in \cite{access2010further} for system simulations of small cells and heterogenous networks. Other default parameters for the simulation are: noise power $\sigma^2 = 10^{-10} \text{W/Hz}$, channel bandwidth $W = 20 \text{M/Hz}$, minimum transmit rate $r_{min} = 100 \text{Mbps}$, maximum delay $d_{max} = 5 \text{sec}$, workload size $\gamma = 1\text{MB}$, user transmit power $P_{tx} = 0.1\text{W}$.

\subsection{Results}
\subsubsection{Performance evaluation}
We compare the performance of E\textsuperscript{2}M\textsuperscript{2} with three benchmark schemes:
\begin{itemize}
\item \textbf{Delay Optimal}: the user always associates with the BS with the lowest delay cost and disregards the energy consumption constraint.
\item \textbf{Energy Optimal}: the user always associates with the BS with the best channel condition without considering the computing performance
\item \textbf{J-step Lookahead}: this is the offline algorithm described in Section IV.B, we set $R = 250$ and $J = 4$. Note that solving the offline problem is extremely computationally complex even if the future information is known.
\end{itemize}

Figure \ref{Performance_E2M2} shows the delay and energy performance of E\textsuperscript{2}M\textsuperscript{2} with FSI and PSI, and the benchmark schemes. As can be seen, \textbf{Delay Optimal} achieves the best delay performance at the cost of violating the energy budget constraint. \textbf{Energy Optimal} uses energy conservatively by always connecting to the BSs with the best channel condition. However, the resulting delay cost is significantly higher than other schemes. Both \textbf{J-step Lookahead} and our two E\textsuperscript{2}M\textsuperscript{2} algorithms satisfy the energy consumption constraint while keeping the delay cost low. Among them, \textbf{J-step Lookahead} is the best as expected. Despite that the proposed E\textsuperscript{2}M\textsuperscript{2} algorithms require no future information and computationally simple, they achieve comparable performance to \textbf{J-step Lookahead}. Moreover, E\textsuperscript{2}M\textsuperscript{2} with PSI is just slightly worse than E\textsuperscript{2}M\textsuperscript{2} with FSI.

\begin{figure}[!htb]
	\centering
	\subfigure[Average delay cost]{\label{average_delay_cost}
		\includegraphics[width=0.4\textwidth]{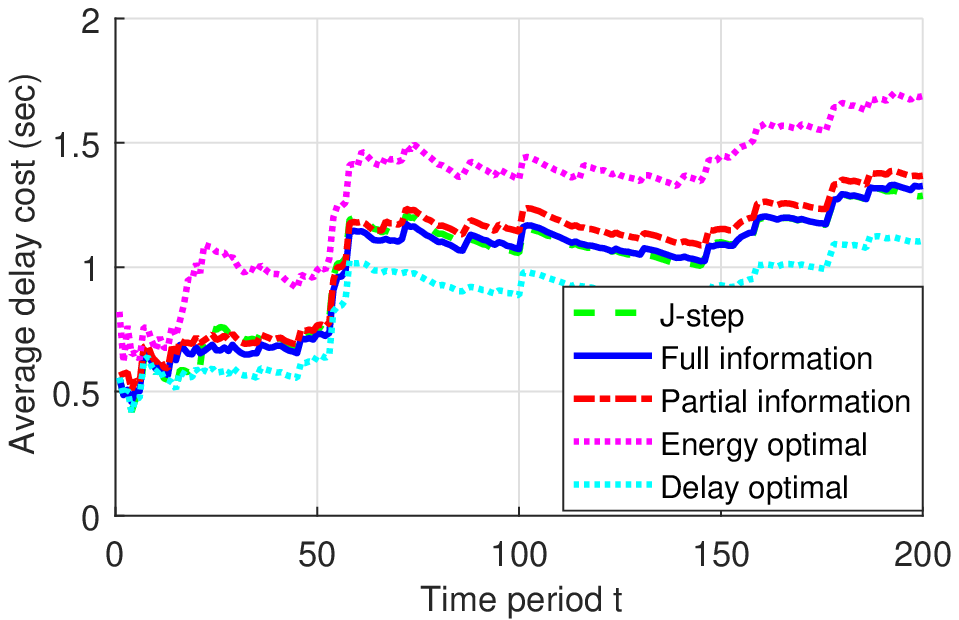}}
	\hspace{0in}
	\subfigure[Total energy consumption]{\label{total_energy_consumption}
		\includegraphics[width=0.4\textwidth]{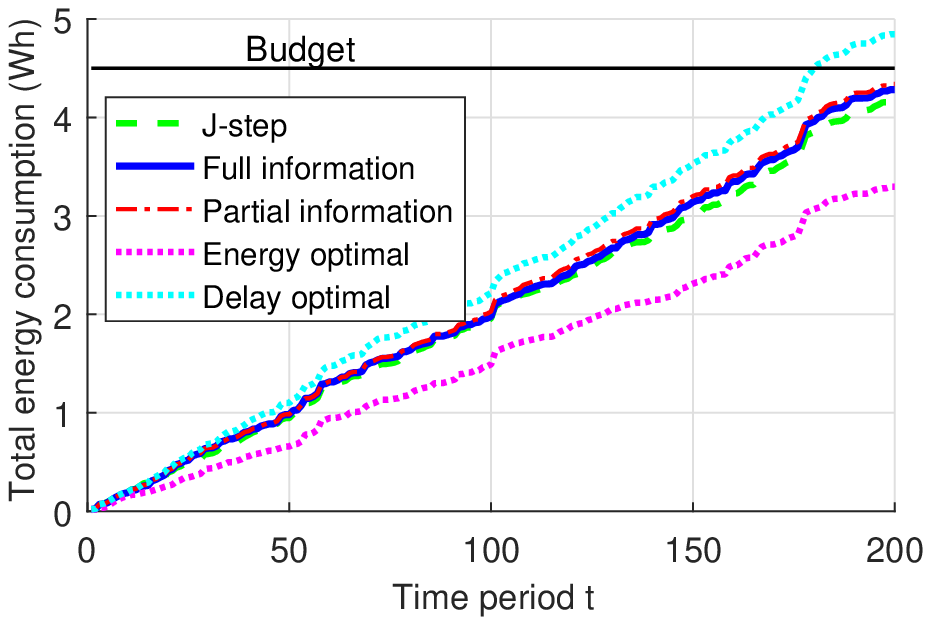}}
	\caption{Performance of E$^2$M$^2$ ($V=0.01$, $\alpha B=120$, $K=100$)}
	\label{Performance_E2M2}
\vspace{-0.2in}
\end{figure}

\subsection{E\textsuperscript{2}M\textsuperscript{2}  with PSI}
We now take a further look at E\textsuperscript{2}M\textsuperscript{2}  with PSI and illustrate the performance of UCB1. Figure \ref{ucb_average_delay_cost} shows that by augmenting E\textsuperscript{2}M\textsuperscript{2} with UCB1, the BS-side information can be learned very quickly and the convergence of the average weighted delay and energy cost to the optimal solution of \textbf{P3} is fast. Figure \ref{switching_times} shows that the number of switchings between BSs is small during the learning process, even when the number of time slots in a period is large. This is much desired as frequent handover is avoided.

\begin{figure}[!htb]
	\centering
	\subfigure[Converence of average delay cost]{\label{ucb_average_delay_cost}
		\includegraphics[width=0.22\textwidth]{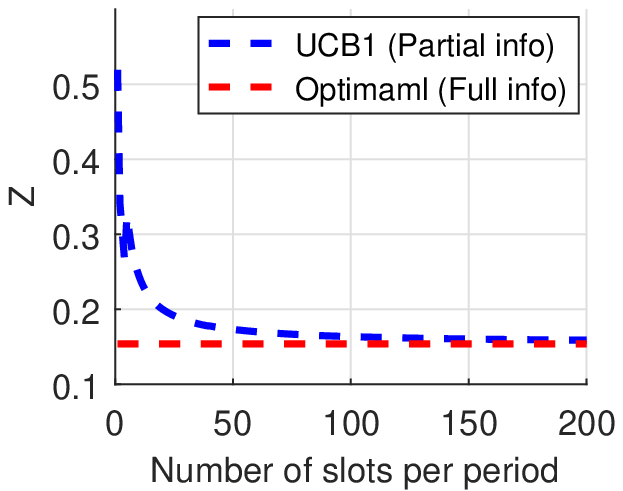}}
	\hspace{0in}
	\subfigure[Switching times]{\label{switching_times}
		\includegraphics[width=0.22\textwidth]{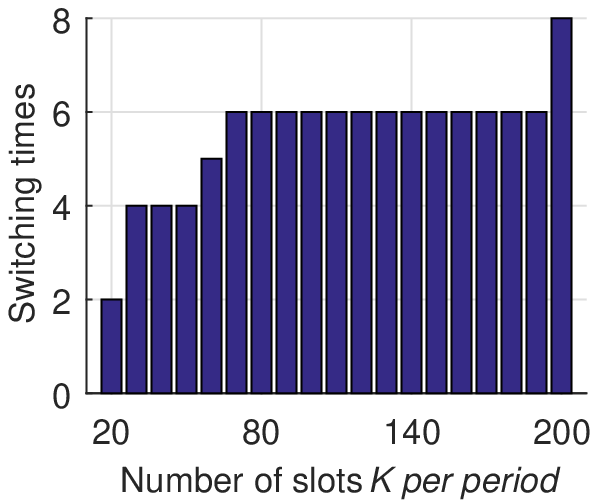}}
	\caption{Dynamics of UCB1}
	\label{dynamics_UCB1}
\end{figure}

\subsection{Impact of energy budget}
Figure \ref{Battery} shows the impact of the energy budget. When the energy budget is large, E\textsuperscript{2}M\textsuperscript{2} achieves the optimal delay performance since the energy constraint is always satisfied. However, it is possible there is no feasible solution when the energy budget is too low, in which case the energy constraint is violated. In between, E\textsuperscript{2}M\textsuperscript{2} makes trade-off between delay cost and energy consumption.

\begin{figure}[!htb]
	\centering
	\subfigure[Average delay cost]{\label{battery_delay_cost}
		\includegraphics[width=0.4\textwidth]{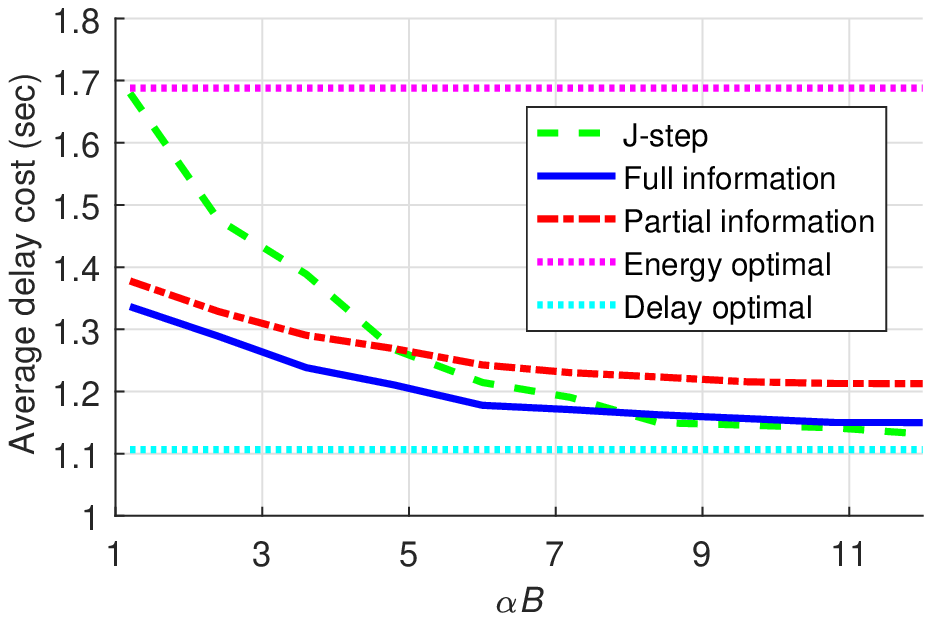}}
	\hspace{0in}
	\subfigure[Total energy consumption]{\label{battery_energy_consumption}
		\includegraphics[width=0.4\textwidth]{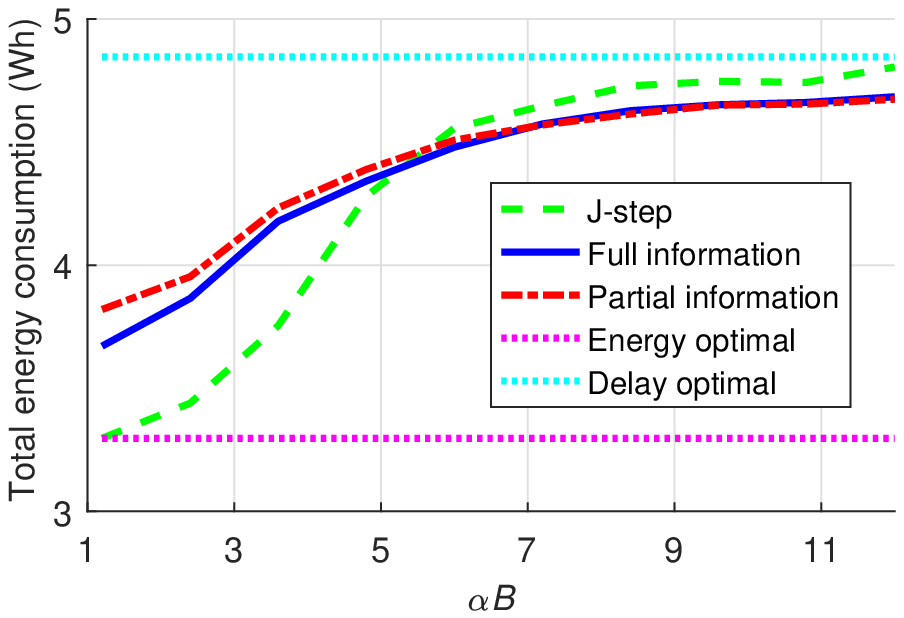}}
	\caption{Impact of battery capacity $\aleph B$}
	\label{Battery}
\vspace{-0.2in}
\end{figure}

\section{Conclusions}
In this paper, we studied the mobility management problem for the next generation mobile networks where the integration of MEC and UDN is envisioned. We developed a novel user-centric mobility management scheme that maximizes the edge computation performance while keeping the user's communication energy consumption low. The proposed scheme works online and effectively handles present and future uncertainties in the system. Future work includes incorporating energy harvesting capabilities into mobile user device and explicitly modeling the handover cost.

\appendix
For notational convenience, we define $y(t) = e(\lambda^t, H^t(a^t, L^t)) - \alpha B/T$ and $d(t) = d(\lambda^t, \mu^t(a^t), s(a^t))$. First, according to the energy deficit queue dynamics, it is easy to see
\begin{align}
q(t+1) - q(t) \geq y(t)
\end{align}
Summing the above over $t = rJ, ..., (r+1)J -1$, using the law of telescoping sums and noting $q(rJ)$ implies
\begin{align}\label{target}
\sum_{t=rJ}^{(r+1)J -1} y(t) \leq q((r+1)J)
\end{align}
where $q((r+1)J)$ is the queue length before reset in period $(r+1)J$. In what follows, we try to bound $q((r+1)J)$.

We define the quadratic Lyapunov function $L(q(t)) \triangleq \frac{1}{2}q^2(t)$.  Squaring the queuing dynamics equation results in the following bound
\begin{align}
q^2(t+1) &\leq  (q(t) + y(t))^2 \nonumber\\
&= q^2(t) + y^2(t) + 2q(t)y(t)
\end{align}
Therefore, the 1-period Lyapunov drift $\Delta_1(t)$ satisifies
\begin{align}
\Delta_1(t) = L(q(t+1)) - L(q(t)) \leq \frac{1}{2} y^2(t) + q(t) y(t)
\end{align}
Now define $U$ as a positive constant that upper bounds $\frac{1}{2} y^2(t)$. Such a constant exists because of the boundedness assumption. Then it can be easily shown that
\begin{align}\label{driftineq}
\Delta_1(t) + V\cdot d(t) \leq U + V\cdot d(t) + q(t)\cdot y(t)
\end{align}
The online algorithm in line 5 of Algorithm 1 actually minimizes the upper bound on the 1-period Lyapunov drift plus a weighted delay cost shown on the right hand side of \eqref{driftineq}. Following \eqref{driftineq}, the $J$-period drift $\Delta_J(rJ) \triangleq L(q((r+1)J)) - L(q(rJ))$ satisfies
\begin{align}
&\Delta_J(rJ) + V \sum_{t=rJ}^{(r+1)J -1}d(t) \\
\leq & UJ + V \sum_{t=rJ}^{(r+1)J -1}d(t) + \sum_{t=rJ}^{(r+1)J-1}q(rJ)\cdot y(t) \nonumber\\
& + \sum_{t=rJ}^{(r+1)J-1} (q(t) - q(rJ))\cdot y(t) \nonumber\\
\leq &UJ(J+1) + V \sum_{t=rJ}^{(r+1)J -1}d(t) + \sum_{t=rJ}^{(r+1)J-1}q(rJ)\cdot y(t) \nonumber
\end{align}
where the second inequality is because $q(t) - q(rJ) \leq (t - rJ) y_\text{max}$, and hence the last term on the right hand side satisfies
\begin{align}
\sum_{t=rJ}^{(r+1)J-1} (q(t) - q(rJ))\cdot y(t) \leq \frac{J^2}{2}y^2_\text{max} \leq J^2 U
\end{align}
By applying E\textsuperscript{2}M\textsuperscript{2} on the left-hand side and considering the optimal $J$-step lookahead algorithm on the right-hand side, we obtain the following
\begin{align}\label{V}
&\Delta_J(rJ) + V_r \sum_{t=rJ}^{(r+1)J -1}d^*(t) \nonumber\\
\leq& UJ(J+1) + V_r J D^*_r + CJ
\end{align}
where $d^*(t)$ is the delay cost achieved by E\textsuperscript{2}M\textsuperscript{2} at time $t$ and $CJ$ is due to the approximate solution of \textbf{P3}.  Therefore,
\begin{align}
&q((r+1)J) = \sqrt{2\Delta_J(rJ)} \nonumber \\
\leq &\sqrt{2UJ(J+1) + CJ + V_r J D^*_r }
\end{align}
Then, by \eqref{target}, we have
\begin{align}
\sum_{t=rJ}^{(r+1)J-1} y(t) \leq \sqrt{2UJ(J+1) + CJ + V_r J D^*_r}
\end{align}
By summing over $r = 0, 1, ..., R-1$ we prove part (2) of Theorem 1.

By dividing both sides of \eqref{V} by $V_r$ and considering $q(rJ) = 0$, it follows that
\begin{align}
\sum_{t=rJ}^{(r+1)J -1}d^*(t) \leq JD^*_r + \frac{UJ(J+1) + CJ}{V_r}
\end{align}
Thus, by summing over $r = 0, 1, ..., R-1$ and dividing both sides by $RJ$, we prove part (1) of Theorem 1.



%
%
%

\bibliographystyle{IEEEtran}
\bibliography{refs}

\end{document}